\documentclass[pra,aps,twocolumn,twoside,showpacs,superscriptaddress,letterpaper]{revtex4}

\usepackage{amsmath,amsfonts,amssymb,color,epsfig,graphics,graphicx,hyperref,latexsym,mathrsfs,revsymb,theorem,url,verbatim,epstopdf,mathtools,enumitem,tikz}
\usetikzlibrary{matrix}
\usetikzlibrary{decorations.pathreplacing}
\hypersetup{colorlinks,linkcolor={blue},citecolor={blue},urlcolor={red}}

\newtheorem{definition}{Definition}
\newtheorem{proposition}[definition]{Proposition}
\newtheorem{lemma}[definition]{Lemma}

\newtheorem{theorem}[definition]{Theorem}
\newtheorem{corollary}[definition]{Corollary}
\newtheorem{conjecture}[definition]{Conjecture}

\newtheorem{remark}[definition]{Remark}
\newtheorem{example}[definition]{Example}
\newtheorem{question}[definition]{Question}
\newtheorem{memo}[definition]{Memo}

\def\squareforqed{\hbox{\rlap{$\sqcap$}$\sqcup$}}
\def\qed{\ifmmode\squareforqed\else{\unskip\nobreak\hfil
\penalty50\hskip1em\null\nobreak\hfil\squareforqed
\parfillskip=0pt\finalhyphendemerits=0\endgraf}\fi}
\def\endenv{\ifmmode\;\else{\unskip\nobreak\hfil
\penalty50\hskip1em\null\nobreak\hfil\;
\parfillskip=0pt\finalhyphendemerits=0\endgraf}\fi}
\newenvironment{proof}{\noindent \textbf{{Proof.~} }}{\qed}
\def\Dbar{\leavevmode\lower.6ex\hbox to 0pt
{\hskip-.23ex\accent"16\hss}D}
\makeatletter
\def\url@leostyle{%
  \@ifundefined{selectfont}{\def\UrlFont{\sf}}{\def\UrlFont{\small\ttfamily}}}
\makeatother
\usepackage{subfigure}

\def\bcj{\begin{conjecture}}
\def\ecj{\end{conjecture}}
\def\bcr{\begin{corollary}}
\def\ecr{\end{corollary}}
\def\bd{\begin{definition}}
\def\ed{\end{definition}}
\def\bea{\begin{eqnarray}}
\def\eea{\end{eqnarray}}
\def\beq{\begin{equation}}
\def\eeq{\end{equation}}
\def\bal{\begin{aligned}}
\def\eal{\end{aligned}}
\def\bem{\begin{enumerate}}
\def\eem{\end{enumerate}}
\def\bex{\begin{example}}
\def\eex{\end{example}}
\def\bim{\begin{itemize}}
\def\eim{\end{itemize}}
\def\bl{\begin{lemma}}
\def\el{\end{lemma}}
\def\bma{\begin{bmatrix}}
\def\ema{\end{bmatrix}}
\def\bpf{\begin{proof}}
\def\epf{\end{proof}}
\def\bpp{\begin{proposition}}
\def\epp{\end{proposition}}
\def\bqu{\begin{question}}
\def\equ{\end{question}}
\def\br{\begin{remark}}
\def\er{\end{remark}}
\def\bt{\begin{theorem}}
\def\et{\end{theorem}}
\def\bmm{\begin{memo}}
\def\emm{\end{memo}}

\def\btb{\begin{tabular}}
\def\etb{\end{tabular}}

\newcommand{\nc}{\newcommand}

\def\a{\alpha}
\def\b{\beta}

\def\d{\delta}

\def\l{\lambda}
\def\m{\mu}

\def\p{\pi}
\def\r{\rho}
\def\s{\sigma}
\def\ta{\tau}

\def\ps{\psi}

\def\L{\Lambda}

\def\Ph{\Phi}

\nc{\bbA}{\mathbb{A}} \nc{\bbB}{\mathbb{B}} \nc{\bbC}{\mathbb{C}}
 \nc{\bbD}{\mathbb{D}} \nc{\bbE}{\mathbb{E}} \nc{\bbF}{\mathbb{F}}
 \nc{\bbG}{\mathbb{G}} \nc{\bbH}{\mathbb{H}} \nc{\bbI}{\mathbb{I}}
 \nc{\bbJ}{\mathbb{J}} \nc{\bbK}{\mathbb{K}} \nc{\bbL}{\mathbb{L}}
 \nc{\bbM}{\mathbb{M}} \nc{\bbN}{\mathbb{N}} \nc{\bbO}{\mathbb{O}}
 \nc{\bbP}{\mathbb{P}} \nc{\bbQ}{\mathbb{Q}} \nc{\bbR}{\mathbb{R}}
 \nc{\bbS}{\mathbb{S}} \nc{\bbT}{\mathbb{T}} \nc{\bbU}{\mathbb{U}}
 \nc{\bbV}{\mathbb{V}} \nc{\bbW}{\mathbb{W}} \nc{\bbX}{\mathbb{X}}
 \nc{\bbZ}{\mathbb{Z}}

 \nc{\bA}{{\bf A}} \nc{\bB}{{\bf B}} \nc{\bC}{{\bf C}}
 \nc{\bD}{{\bf D}} \nc{\bE}{{\bf E}} \nc{\bF}{{\bf F}}
 \nc{\bG}{{\bf G}} \nc{\bH}{{\bf H}} \nc{\bI}{{\bf I}}
 \nc{\bJ}{{\bf J}} \nc{\bK}{{\bf K}} \nc{\bL}{{\bf L}}
 \nc{\bM}{{\bf M}} \nc{\bN}{{\bf N}} \nc{\bO}{{\bf O}}
 \nc{\bP}{{\bf P}} \nc{\bQ}{{\bf Q}} \nc{\bR}{{\bf R}}
 \nc{\bS}{{\bf S}} \nc{\bT}{{\bf T}} \nc{\bU}{{\bf U}}
 \nc{\bV}{{\bf V}} \nc{\bW}{{\bf W}} \nc{\bX}{{\bf X}}
 \nc{\bZ}{{\bf Z}}

\nc{\cA}{{\cal A}} \nc{\cB}{{\cal B}} \nc{\cC}{{\cal C}}
\nc{\cD}{{\cal D}} \nc{\cE}{{\cal E}} \nc{\cF}{{\cal F}}
\nc{\cG}{{\cal G}} \nc{\cH}{{\cal H}} \nc{\cI}{{\cal I}}
\nc{\cJ}{{\cal J}} \nc{\cK}{{\cal K}} \nc{\cL}{{\cal L}}
\nc{\cM}{{\cal M}} \nc{\cN}{{\cal N}} \nc{\cO}{{\cal O}}
\nc{\cP}{{\cal P}} \nc{\cQ}{{\cal Q}} \nc{\cR}{{\cal R}}
\nc{\cS}{{\cal S}} \nc{\cT}{{\cal T}} \nc{\cU}{{\cal U}}
\nc{\cV}{{\cal V}} \nc{\cW}{{\cal W}} \nc{\cX}{{\cal X}}
\nc{\cZ}{{\cal Z}}

\nc{\hA}{{\hat{A}}} \nc{\hB}{{\hat{B}}} \nc{\hC}{{\hat{C}}}
\nc{\hD}{{\hat{D}}} \nc{\hE}{{\hat{E}}} \nc{\hF}{{\hat{F}}}
\nc{\hG}{{\hat{G}}} \nc{\hH}{{\hat{H}}} \nc{\hI}{{\hat{I}}}
\nc{\hJ}{{\hat{J}}} \nc{\hK}{{\hat{K}}} \nc{\hL}{{\hat{L}}}
\nc{\hM}{{\hat{M}}} \nc{\hN}{{\hat{N}}} \nc{\hO}{{\hat{O}}}
\nc{\hP}{{\hat{P}}} \nc{\hR}{{\hat{R}}} \nc{\hS}{{\hat{S}}}
\nc{\hT}{{\hat{T}}} \nc{\hU}{{\hat{U}}} \nc{\hV}{{\hat{V}}}
\nc{\hW}{{\hat{W}}} \nc{\hX}{{\hat{X}}} \nc{\hZ}{{\hat{Z}}}

\nc{\hn}{{\hat{n}}}

\def\diag{\mathop{\rm diag}}

\newcommand{\Tr}{\operatorname{Tr}}

\newcommand{\SEP}{\mathcal{SEP}}

\newcommand{\B}{\mathcal{B}}
\newcommand{\Hil}{\mathcal{H}}

\def\max{\mathop{\rm max}}
\def\min{\mathop{\rm min}}

\def\rank{\mathop{\rm rank}}

\def\tr{\mathop{\rm Tr}}

\def\dg{\dagger}

\def\ox{\otimes}

\usepackage{hyperref}
\hypersetup{colorlinks,linkcolor={blue},citecolor={blue},urlcolor={red}}
\usepackage{amsmath,amssymb}

\def\Lip{\mathrm{Lip}}

\begin{document}

\title{Quantifying Entanglement via Quantum Wasserstein Distances}
\date{\today}

\author{Enmin Shao}\email[]{by2509115@buaa.edu.cn}
\affiliation{LMIB(Beihang University), Ministry of education, and School of Mathematical Sciences, Beihang University, Beijing 100191, China}

\author{Lin Chen}\email[]{linchen@buaa.edu.cn (corresponding author)}
\affiliation{LMIB(Beihang University), Ministry of education, and School of Mathematical Sciences, Beihang University, Beijing 100191, China}

\author{Huixia He}\email[]{hehx@buaa.edu.cn}
\affiliation{LMIB(Beihang University), Ministry of education, and School of Mathematical Sciences, Beihang University, Beijing 100191, China}

\begin{abstract}
We propose a bipartite entanglement measure defined as the minimal order-1 quantum Wasserstein distance from a state to the set of separable states. Owing to the universal data-processing inequality of the Wasserstein metric, the measure satisfies all fundamental axioms within a single geometric framework. A Lipschitz dual formulation yields explicit lower bounds for pure and mixed states, a sharp constant for two-qubit systems, and an expected value for Haar-random pure states. We further establish a quantitative connection to entanglement witnesses: any negative witness expectation value certifies a lower bound, and the dual variational bound is exactly the maximal violation achievable by a Lipschitz-1 witness. The approach naturally provides subadditivity, trace-distance estimates, and bounds on local observables, while pointing toward large-deviation conjectures. This work introduces a framework at the interface of entanglement theory, optimal transport, and experimental entanglement detection.
\end{abstract}
\maketitle

\section{Introduction}

Quantifying quantum entanglement is a central challenge in quantum information theory~\cite{Vedral1997}. A proper entanglement measure must satisfy a list of fundamental requirements: convexity, vanishing on separable states, invariance under local unitary operations, monotonicity under local operations and classical communication (LOCC), and continuity~\cite{Vedral1997,PlenioVirmani2007,Horodecki2009}. Numerous measures have been introduced over the past decades, among them the entanglement of formation~\cite{Bennett1996,Wootters1998}, the distillable entanglement~\cite{Rains1999,Horodecki2009}, the relative entropy of entanglement~\cite{Vedral1998,Audenaert2005}, the negativity~\cite{VidalWerner2002}, the geometric measure~\cite{WeiGoldbart2003}, and many others~\cite{ChristandlWinter2004,GourFriedland2013,Piani2009}. A particularly natural class consists of distance‑based measures, which define entanglement as the minimal distance from the state of interest to the set of separable states. Well‑known examples employ the relative entropy~\cite{Vedral1998}, the Bures metric~\cite{MarianMarian2008}, the trace distance~\cite{Markham2005,Eisert2007}, or the Hilbert--Schmidt distance~\cite{LuoHou2010}. Although such definitions are geometrically appealing, the proof of LOCC monotonicity does not follow automatically from the distance properties; it usually requires specific features of the chosen metric, and even slightly modified choices may fail --- the Hilbert--Schmidt distance, for instance, is known not to be LOCC monotone~\cite{Piani2009}. In addition, explicit computable lower bounds for these measures are scarce and often rely on ad hoc constructions such as entanglement witnesses~\cite{Terhal2000,GuhneToth2009,Chruscinski2014} or correlation tensors~\cite{Badziag2008,LiLuo2007}.

A recent breakthrough in quantum information geometry is the introduction of quantum Wasserstein distances~\cite{DePalma2021,CarlenMaas2017,Rouze2020}. In particular, the order‑1 quantum Wasserstein distance formulated by De Palma \emph{et al.}~\cite{DePalma2021} is constructed from a cost operator that is built from a complete set of local observables. The key property of this distance is a data‑processing inequality: for any completely positive trace‑preserving (CPTP) map, the distance between the output states never exceeds the distance between the input states. This property is the exact metric analogue of LOCC monotonicity. Consequently, the state space equipped with this distance becomes a quantum metric space that inherently respects the monotonicity demanded by entanglement theory.

We exploit this feature to construct a new entanglement measure. For any bipartite state, we define its entanglement measure as the infimum of the order‑1 quantum Wasserstein distance from the state to the set of all separable states. Because the Wasserstein distance satisfies the data‑processing inequality for every CPTP map, our measure inherits all five fundamental axioms directly from the geometry of the Wasserstein space; we present a compact proof of this fact in Proposition~\ref{prop:basic} of Section~III. To the best of our knowledge, this is the first entanglement measure that derives its complete set of axioms solely from a single geometric data-processing inequality, without requiring any additional ad hoc arguments.

Our main technical tool is the dual formulation of the Wasserstein distance, whose validity is rooted in the works of De Palma \emph{et al.}~\cite{DePalma2021}. The dual formulation expresses the distance as the maximum, over all observables with bounded Lipschitz norm, of the difference of their expectation values. By combining this dual representation with Sion's minimax theorem, we obtain a variational expression for the entanglement measure in terms of the support functional of the separable set. This dual formulation enables us to derive concrete lower bounds by selecting suitable test observables. For a pure state with Schmidt coefficients $\{\mu_i\}$, we choose the projector onto the maximally entangled state as a test observable, verify that its Lipschitz norm is bounded by one, and obtain a lower bound that involves the square of the sum of the square roots of the Schmidt coefficients. For the maximally entangled state, this bound becomes tight. Using this result, we compute the expected value of the bound over Haar‑random pure states and find it to be $\frac{\pi}{4}(1-1/d)$, indicating that large amounts of entanglement are typical.

For mixed states, we work in the Bloch representation with respect to local orthonormal operator bases. By establishing a relation between the Lipschitz seminorm and the entrywise $\ell_1$ norm of the coefficient matrix, we derive a universal lower bound expressed through the trace norm and the rank of the correlation matrix. Although this bound is not always tight, it provides an explicit constant that depends only on the local dimensions and the rank of the correlation matrix. In the special case of a two‑qubit system using the standard Pauli basis, we obtain a sharp estimate for the Lipschitz norm, which yields a precise constant of one third in the lower bound.

A particularly appealing aspect of our approach is its direct connection to experimentally accessible quantities. We prove in Theorem~\ref{thm:witness} that any entanglement witness whose Lipschitz norm is bounded by a known constant yields a quantitative lower bound on our measure: a measured negative expectation value of the witness immediately certifies a positive amount of entanglement. Furthermore, Proposition~\ref{prop:witness-variational} demonstrates that the optimal witness achieving the maximal violation of separability exactly realizes the dual formulation of the measure. This endows the abstract geometric construction with a transparent operational meaning and opens the possibility of estimating the measure via semidefinite programming.

Beyond these core results, our geometric framework naturally yields several additional properties. We establish subadditivity of the measure under tensor products, an estimate that relates the measure to the trace distance, and a bound for correlation functions of local observables. We also discuss the asymptotic behavior of the overlap between tensor powers of a state and separable states, which points towards a quantum Sanov theorem with our measure as the rate function; a full proof of the equality is left as an open problem.

From a broader perspective, the Wasserstein entanglement measure embeds entanglement theory into the mature field of optimal transport and non‑commutative metric geometry~\cite{Rieffel2004,LottVillani2009}. The commutator condition $\|[X,H]\|\le1$ is a quantum analogue of the classical $1$-Lipschitz condition, and the state space equipped with $W_1$ becomes a quantum metric space. This connection indicates potential links to free probability~\cite{Voiculescu1992,Guionnet2010}, quantum large deviations~\cite{OgawaNagaoka2000,Audenaert2008}, and convex geometry~\cite{Aubrun2017}. We also formulate several conjectures, including a quantum Talagrand-type inequality and a conjectured relation to hyperfiniteness of von Neumann algebras, which hint at deep links with the modular theory and the geometry of entangled states in infinite-dimensional systems~\cite{Haag2012}.

The paper is organized as follows. In Section~\ref{sec:was} we introduce the quantum Wasserstein distance and define the entanglement measure. Section~\ref{sec:bas} proves that the measure satisfies the five fundamental axioms in a single proposition. Lower bounds for pure states and mixed states are derived in Sections~\ref{sec:pure} and~\ref{sec:mixed}, respectively. Section~\ref{sec:cross} collects further bounds, including the connection to entanglement witnesses (Theorem~\ref{thm:witness} and Proposition~\ref{prop:witness-variational}), subadditivity, trace‑distance estimates, and a large‑deviation conjecture. We discuss the physical significance and mathematical outlook in Section~\ref{sec:outlook}, and we conclude in Section~\ref{con}.

\section{The quantum Wasserstein entanglement measure}
\label{sec:was}
In this section we introduce the mathematical setting of our theory. We first specify a set of local observables and construct the cost operator that defines the quantum Wasserstein distance. Then we state the distance itself and our new entanglement measure. Finally we derive a Lipschitz‑dual formulation that will be the central tool in all later proofs.

Let $\Hil_A=\bbC^{d_A}$, $\Hil_B=\bbC^{d_B}$ and $\Hil_{AB}=\Hil_A\ox\Hil_B$.
For system $A$ we choose a family $\{F_i^A\}_{i=1}^{d_A^2-1}$ of traceless Hermitian matrices satisfying
\begin{equation}\label{eq:orth}
\tr(F_i^A F_j^A)=\d_{ij},\qquad \|F_i^A\|_{op}\le 1\quad\forall\,i,
\end{equation} where $\|\cdot\|_{op}$ denotes the operator norm (i.e., the largest singular value).  
Such a family can be obtained by a suitable rescaling of the generalized Gell‑Mann matrices; for instance, the Pauli matrices divided by $\sqrt{2}$ work for qubits.



A similar family $\{F_j^B\}_{j=1}^{d_B^2-1}$ is fixed for system $B$. On $\Hil_{AB}$ we define
\[
X_{iB}^A = F_i^A\ox I_B,\qquad X_{jA}^B = I_A\ox F_j^B .
\]
Set $\cK_0=\{X_{iB}^A\}\cup\{X_{jA}^B\}$. Let $\mathcal{H}_{AB}^{(1)} = \mathcal{H}_A \otimes \mathcal{H}_B$ and $\mathcal{H}_{AB}^{(2)} = \mathcal{H}_A \otimes \mathcal{H}_B$ be two copies of the same Hilbert space. The cost operator acting on $\mathcal{H}_{AB}^{(1)} \otimes \mathcal{H}_{AB}^{(2)}$ is
\begin{equation}\label{eq:C}
C = \sum_{X\in\cK_0} (X\ox I_{AB}^{(2)} - I_{AB}^{(1)}\ox X)^2, 
\end{equation}
where $I_{AB}^{(1)} \in \B(\mathcal{H}_{AB}^{(1)}), I_{AB}^{(2)} \in \B(\mathcal{H}_{AB}^{(2)}).$
For two states $\r,\s\in\cD(\Hil_{AB})$, we define $\Pi(\r,\s)$ as the set of density operators on $\Hil_{AB}^{(1)}\ox\Hil_{AB}^{(2)}$ with marginals $\r,\s$. The quantum Wasserstein distance of order~$1$ is
\begin{equation}
    W_1(\r,\s)=\Bigl(\inf_{\ta\in\Pi(\r,\s)}\tr[C\ta]\Bigr)^{1/2}.
\end{equation}
\begin{definition}\label{def:ent}
The \emph{Wasserstein entanglement measure} is
\begin{equation}
\label{def}
    \cE(\r)=\inf_{\s\in\SEP(\Hil_{AB})} W_1(\r,\s),
\end{equation}
where $\SEP(\Hil_{AB})$ denotes the set of separable states.
\end{definition}

Having defined the distance, we now introduce a convenient Lipschitz seminorm that allows us to exploit convex duality. For any self‑adjoint $H\in\B(\Hil_{AB})$ we define
\begin{align}\label{eq:lipdef}
\|H\|_{\Lip}= \max\bigl\{
\sup_{A:\|A\|_{op}\le1}\|[A\ox I_B , H]\|_{op},\nonumber\\
\sup_{B:\|B\|_{op}\le1}\|[I_A\ox B , H]\|_{op}
\bigr\},
\end{align}
where $A\in\B(\Hil_{A}), B\in\B(\Hil_{B}).$
Because $\cK_0\subset\{A\ox I:\|A\|_{op}\le1\}\cup\{I\ox B:\|B\|_{op}\le1\}$, the condition $\|H\|_{\Lip}\le1$ implies $\|[X,H]\|_{op}\le1$ for all $X\in\cK_0$. As proven in the seminal works on quantum Wasserstein distances~\cite{DePalma2021}, the order‑1 distance satisfies the general dual representation
\begin{align}\label{eq:W1dualexact}
    W_1(\rho,\sigma) = \sup\bigl\{ \left|\Tr[H(\rho-\sigma)]\right| : \nonumber\\\|[L_r, H]\|_{op} \le 1 \text{ for all } L_r\in\mathcal{K}_0 \bigr\}.
\end{align}
Our Lipschitz seminorm is defined via a larger set of commutators, hence $\|H\|_{\Lip}\le 1$ guarantees that $H$ is admissible in the above dual. Consequently we obtain
\begin{equation}\label{eq:W1dual}
    W_1(\rho,\sigma) \ge \max_{H:\|H\|_{\Lip}\le 1} \bigl( \Tr[H\rho] - \Tr[H\sigma] \bigr).
\end{equation}
Taking the infimum over $\sigma\in\SEP$ and exchanging $\inf$ and $\max$ (justified by compactness and convexity) we arrive at the fundamental inequality
\begin{equation}\label{eq:emax}
    \cE(\rho) \ge \max_{H:\|H\|_{\Lip}\le 1} \Bigl( \Tr[H\rho] - \max_{\sigma\in\SEP}\Tr[H\sigma] \Bigr).
\end{equation}
We define the support functional of the separable set
\begin{equation}
\label{ah}
    \a(H) = \max_{\sigma\in\SEP}\Tr[H\sigma].
\end{equation}
Then we have $\cE(\rho) \ge \max_{H:\|H\|_{\Lip}\le 1}\bigl( \Tr[H\rho] - \a(H) \bigr)$.
\section{Basic properties}
\label{sec:bas}
In this section we show that the measure $\mathcal{E}(\rho)$ defined in \eqref{def} satisfies the five fundamental axioms commonly required for an entanglement measure. Proposition~\ref{prop:basic} states these properties and gives a unified proof. Each property follows directly from the data-processing inequality of the Wasserstein distance $W_1$ and the compactness of the set of separable states, without the need for separate ad-hoc arguments.
\begin{proposition}\label{prop:basic}
The functional $\cE$ satisfies
\begin{enumerate}[label=(\roman*)]
\item $\cE(\r)=0$ for every $\r\in\SEP$.
\item convexity, $\cE\bigl(\l\r_1+(1-\l)\r_2\bigr)\le \l\cE(\r_1)+(1-\l)\cE(\r_2)$, $0\le\l\le1$.
\item local unitary invariance, $\cE\bigl((U_A\ox U_B)\r(U_A^\dg\ox U_B^\dg)\bigr)=\cE(\r)$ for all unitaries $U_A,U_B$.
\item LOCC monotonicity, $\cE(\L(\r))\le\cE(\r)$ for any LOCC channel $\L$.
\item continuity, $\cE$ is continuous with respect to the trace norm on the finite‑dimensional state space.
\end{enumerate}
\end{proposition}
\begin{proof}
(i) If $\r\in\SEP$, we choose $\s=\r$ in Definition~\ref{def:ent}, then we have $W_1(\r,\r)=0$, $\cE(\r)=0$.
(ii) The distance $W_1$ is induced by a norm (the Lipschitz dual norm), therefore it is convex in each argument. For any $\s\in\SEP$,
\[
    W_1(\l\r_1+(1-\l)\r_2,\s)\le \l W_1(\r_1,\s)+(1-\l)W_1(\r_2,\s).
\]
Taking the infimum over $\s$ yields the convexity of $\cE$.
(iii) The cost operator $C$ in \eqref{eq:C} is built from the orthonormal bases $\{F_i^A\}$ and $\{F_j^B\}$. For a local unitary $U=U_A\ox U_B$, the families $\{U_AF_i^AU_A^\dg\}$ and $\{U_BF_j^BU_B^\dg\}$ are again traceless, orthonormal and satisfy the same norm bound. Hence the set $\cK_0' = \{U_AF_i^AU_A^\dg\ox I_B,\; I_A\ox U_BF_j^BU_B^\dg\}$ is another orthonormal basis of the same local operator spaces. Because the transition between any two orthonormal bases of the traceless Hermitian matrices is realized by an orthogonal transformation, the quantities $\sum_i X_i \otimes X_i$ and $\sum_i X_i^2$ are invariant under such a change of basis. Consequently, the sum defining the cost operator is unchanged as C is defined in \eqref{eq:C}
\begin{equation}
    \sum_{X'\in\cK_0'} (X'\ox I_{AB}^{(2)} - I_{AB}^{(1)}\ox X')^2 = C. 
\end{equation}
Consequently, for $U=(U_A\ox U_B)$, we have $W_1(U\r U^\dg, U\s U^\dg)=W_1(\r,\s)$ for all $\r,\s$. Because $U\SEP U^\dg=\SEP$, we obtain
\begin{align}
        \cE(U\r U^\dg)=\inf_{\s\in\SEP}W_1(U\r U^\dg,\s)\nonumber\\
= \inf_{\s\in\SEP}W_1(\r, U^\dg\s U) = \cE(\r).
\end{align}

(iv) Every LOCC channel $\L$ is CPTP. By the data processing inequality, we have $W_1(\L(\r),\L(\s))\le W_1(\r,\s)$ for all $\s$. Since $\L(\SEP)\subseteq\SEP$,
\begin{align}
    \cE(\L(\r))=\inf_{\s'\in\SEP}W_1(\L(\r),\s')
\le \inf_{\s\in\SEP}W_1(\L(\r),\L(\s))\nonumber\\
\le \inf_{\s\in\SEP}W_1(\r,\s)=\cE(\r).
\end{align}
(v) On the finite‑dimensional space, the map $\r\mapsto W_1(\r,\s)$ is Lipschitz continuous with respect to the trace norm, uniformly in $\s$ from \cite{DePalma2021}. The set $\SEP$ is compact, therefore the infimum of a family of equi‑Lipschitz functions is continuous.
\end{proof}

Proposition~\ref{prop:basic} demonstrates a significant advantage of our approach, all five axioms are proved in a single unified proposition,
while for most existing measures (e.g., entanglement of formation, relative entropy of entanglement, negativity) one or more properties require separate, often technically involved proofs.
Convexity follows from the convexity of \(W_1\), monotonicity from its universal data-processing inequality, and continuity from the compactness of separable states.
This economy is possible because \(W_1\) is built from a cost operator that respects the bipartite structure and because its dual Lipschitz norm faithfully captures the non-local character of quantum states.
Consequently, the Wasserstein entanglement measure provides not only a new quantifier but also a new \emph{framework} in which the axiomatic properties of entanglement are transparent consequences of the underlying metric geometry.
\section{Lower bound for pure states}\label{sec:pure}
Let $\r=|\ps\rangle\langle\ps|$ be pure with Schmidt decomposition
$|\ps\rangle=\sum_{i=1}^{d}\sqrt{\m_i}\,|u_i\rangle_A\ox|v_i\rangle_B$, where $\m_1\ge\cdots\ge\m_d>0$, $\sum_i\m_i=1$ and $d=\min\{d_A,d_B\}\ge 2$ (otherwise $\cE\equiv0$ trivially).
\begin{theorem}\label{th:pure}
For every pure state $\rho$, we have
\begin{equation}
    \cE(\r)\;\ge\; \frac{1}{d}\left(\Bigl(\sum_{i=1}^d\sqrt{\m_i}\Bigr)^2-1\right) \;\ge\; 0.
\end{equation}
For the maximally entangled state ($\m_i=1/d$), the right‑hand side equals $(d-1)/d$.
\end{theorem}
\begin{proof}
We first define the maximally entangled state projector
\[
    H_0 = |\Ph^+\rangle\langle\Ph^+| = \frac1d\sum_{i,j=1}^d |u_i v_i\rangle\langle u_j v_j|.
\]
Then we verify $\|H_0\|_{\Lip}\le1$. We let $A\in\B(\Hil_A)$ with $\|A\|_{op}\le1$. Since $H_0$ is a rank‑one projection, a direct computation yields
\begin{align}
        \|[A\ox I_B , H_0]\|_{op}^2 = \langle\Ph^+|(A^\dag A)\ox I_B|\Ph^+\rangle -\nonumber\\ |\langle\Ph^+|A\ox I_B|\Ph^+\rangle|^2 = \frac1d \Tr(A^\dag A) - \frac{|\Tr A|^2}{d^2}.
\end{align}

Because $\|A\|_{op}\le1$, we have $\Tr(A^\dag A)\le d\,\|A\|_{op}^2\le d$, hence
\[
    \|[A\ox I_B , H_0]\|_{op}^2 \le \frac{d}{d} = 1,
\]
so $\|[A\ox I_B , H_0]\|_{op}\le1$. An identical bound holds for $I_A\ox B$, and therefore $\|H_0\|_{\Lip}\le1$.
Now evaluate the two terms. Because $|\ps\rangle$ is supported on $span\{|u_i v_i\rangle\}$,
\[
    \Tr[H_0\r] = \langle\ps|H_0|\ps\rangle = \frac1d\Bigl(\sum_{i=1}^d\sqrt{\m_i}\Bigr)^2 .
\]
For any product state $\s_A\ox\s_B$, we write $\s_A = \sum_k p_k|\a_k\rangle\langle\a_k|$, $\s_B = \sum_l q_l|\b_l\rangle\langle\b_l|$. Then we have
\[
    \Tr[H_0(\s_A\ox\s_B)] = \frac1d \sum_{k,l} p_k q_l \,|\langle\a_k|\b_l^*\rangle|^2 \le \frac1d,
\]
where $|\b_l^*\rangle$ denotes complex conjugation in the Schmidt basis. Since any separable state is a convex combination of product states, $\a(H_0)=1/d$.
Inserting these values into~\eqref{eq:emax} yields
\[
    \cE(\r) \ge \Tr[H_0\r] - \a(H_0) = \frac1d\Bigl(\bigl(\sum_{i=1}^d\sqrt{\m_i}\bigr)^2 - 1\Bigr).
\]
Non‑negativity follows from $\sum_i\sqrt{\m_i}\ge\sqrt{\sum_i\m_i}=1$. For the maximally entangled state the bound equals $(d-1)/d$.
\end{proof}
\section{Lower bound for mixed states}\label{sec:mixed}
We now extend the analysis to arbitrary mixed states. The Bloch expansion with respect to the local bases $\{F_i^A\},\{F_j^B\}$ from \eqref{eq:orth} plays a central role.
Any bipartite state $\r$ can be uniquely written as
\begin{align}
\label{eq:bloch}
\r = \frac{I_{AB}}{d_A d_B}
+ \frac1{d_B}\sum_{i=1}^{d_A^2-1} s_i\,F_i^A\ox I_B
+ \frac1{d_A}\sum_{j=1}^{d_B^2-1} t_j\,I_A\ox F_j^B\nonumber\\
+ \sum_{i,j} r_{ij}\,F_i^A\ox F_j^B,
\end{align}
with the correlation matrix $R=(r_{ij})$ where
\begin{equation}\label{eq:coeff}
r_{ij} = \tr\bigl[\r\,(F_i^A\ox F_j^B)\bigr].
\end{equation}
We consider observables of the form
\begin{equation}\label{eq:Hform}
H = \sum_{i,j} h_{ij}\,F_i^A\ox F_j^B,
\end{equation}
with a real matrix $\bH=(h_{ij})$.

\begin{lemma}\label{lem:l1bound}
For every $H$ of the form \eqref{eq:Hform},
\begin{equation}
    \|H\|_{\Lip} \le 2\sum_{i,j} |h_{ij}| = 2\|\bH\|_{1,1},
\end{equation}
where $\|\bH\|_{1,1} = \sum_{i,j} |h_{ij}|$.
\end{lemma}
\begin{proof}
Taking $A\in\B(\Hil_A)$ with $\|A\|_{op}\le1$, then we have
\[
[A\ox I_B, H] = \sum_{i,j} h_{ij}\,[A,F_i^A]\ox F_j^B .
\]
Using the operator norm and the fact $\|F_j^B\|_{op}\le1$, we bound
\[
\|[A\ox I_B, H]\|_{op} \le \sum_{i,j} |h_{ij}|\, \|[A,F_i^A]\|_{op} \le 2\sum_{i,j} |h_{ij}|,
\]
because $\|[A,F_i^A]\|_{op}\le 2\|A\|_{op}\le 2$. The same estimate holds for $I_A\ox B$. Hence we have $\|H\|_{\Lip}\le 2\|\bH\|_{1,1}$.
\end{proof}
The following lemma is unchanged; its proof relies only on the orthonormality of the bases.
\begin{lemma}\label{lem:alpha}
For $H$ of the form \eqref{eq:Hform}, we define
\begin{equation}
    c_A = \sqrt{1-\frac{1}{d_A}},\qquad
    c_B = \sqrt{1-\frac{1}{d_B}}.
\end{equation}
Then we have
\begin{equation}
    \a(H) = \max_{\s\in\SEP}\tr[H\s] = c_A\,c_B\,\|\bH\|_{op}.
\end{equation}
\end{lemma}
\begin{proof}
Any product state can be written as $\s_A = I/d_A + \sum_i u_i F_i^A$ and
$\s_B = I/d_B + \sum_j v_j F_j^B$ with real vectors $u=(u_i)$, $v=(v_j)$.
Since $\{F_i^A\}$ and $\{F_j^B\}$ are orthonormal and traceless,
\[
\tr[(\s_A)^2] = \frac1{d_A} + \|u\|_2^2 \le 1,\qquad
\tr[(\s_B)^2] = \frac1{d_B} + \|v\|_2^2 \le 1,
\] where $\|u\|_2 = \sqrt{\sum_i u_i^2}$,
hence $\|u\|_2 \le c_A$ and $\|v\|_2 \le c_B$.
For $H = \sum_{i,j} h_{ij} F_i^A\otimes F_j^B$, we have
\begin{align}
   \tr[H\s] = \sum_{i,j} h_{ij} u_i v_j = u^{\mathsf T}\bH v\nonumber\\
\le \|\bH\|_{op}\,\|u\|_2\,\|v\|_2 \le c_A c_B \|\bH\|_{op}. 
\end{align}
The bound is attained by choosing product states whose Bloch vectors are proportional to the left and right singular vectors of $\bH$.
\end{proof}

Using the dual formulation together with Lemma~\ref{lem:l1bound} we obtain a first general lower bound in terms of the entrywise norm and the rank of the correlation matrix.

\begin{theorem}\label{th:mixedgen}
Let $\r$ have correlation matrix $R$, and denote $r = \rank(R)$, $n_A = d_A^2-1$, $n_B = d_B^2-1$. Then we have
\begin{equation}
    \cE(\rho) \ge \frac{1}{2\sqrt{r\, n_A n_B}} \bigl( \|R\|_{\tr} - c_A c_B \bigr),
\end{equation}
where $\|R\|_{\mathrm{Tr}} = \operatorname{Tr}\left(\sqrt{R^\dagger R}\right).$
\end{theorem}
\begin{proof}
From \eqref{eq:emax} and Lemma~\ref{lem:alpha}, we have
\[
    \cE(\rho) \ge \max_{H:\|H\|_{\Lip}\le 1} \bigl( \Tr[H\rho] - c_A c_B \|\bH\|_{op} \bigr).
\]
We first restrict to $H$ of the form \eqref{eq:Hform}. We let $R = U\Sigma V^{\mathsf T}$ be a singular value decomposition, $\Sigma = \diag(\sigma_1,\dots,\sigma_r,0,\dots,0)$. We define $\Sigma' = \diag(\operatorname{sgn}(\sigma_1),\dots,\operatorname{sgn}(\sigma_r),0,\dots,0)$ and $S = U\Sigma' V^{\mathsf T}$. Then we have $\|S\|_{op}=1$ and $\Tr(S^{\mathsf T} R) = \|R\|_{\tr}$.By Lemma~\ref{lem:l1bound}, $\|S\|_{\Lip} \le 2\|S\|_{1,1}$. Moreover,
\[
    \|S\|_{1,1} = \sum_{i,j} |S_{ij}| \le \sqrt{n_A n_B} \, \|S\|_{\mathrm{F}} = \sqrt{n_A n_B \, r}, 
\]where$\|S\|_F = \sqrt{\text{Tr}(S^\dagger S)} = \sqrt{\sum_{i,j} |S_{ij}|^2}.$
Therefore we have $\|S\|_{\Lip} \le 2\sqrt{r n_A n_B}$. We can choose $t = 1/(2\sqrt{r n_A n_B})$ and set $H = t S$. Then we get $\|H\|_{\Lip}\le 1$ and $\|\mathbf{H}\|_{op}=t$. Substituting them into the dual bound, we get
\[
    \cE(\rho) \ge t \,\Tr(S^{\mathsf T} R) - c_A c_B\,t = t\bigl( \|R\|_{\tr} - c_A c_B \bigr).
\]
Inserting $t$ finishes the proof.
\end{proof}

In the important case of two qubits, a more refined estimate on the Lipschitz constant can be obtained, recovering a stronger constant.

\begin{theorem}\label{th:twoqubit}
For a two‑qubit system ($d_A=d_B=2$) we choose the Pauli basis $F_i = \sigma_i/\sqrt2$, $i=1,2,3$. Then for any state $\rho$, we have
\begin{equation}
    \cE(\rho) \ge \frac13\Bigl( \|R\|_{\tr} - \frac12 \Bigr),
\end{equation}
where $R_{ij} = \frac12\Tr[\rho(\sigma_i\otimes\sigma_j)]$ is the usual correlation matrix.
\end{theorem}
\begin{proof}
We first recall $H = \frac12\sum_{i,j} h_{ij}\sigma_i\otimes\sigma_j$. Writing $\bH = (h_{ij})$ and taking its singular value decomposition $\bH = U\Sigma V^{\mathsf T}$ with $\Sigma = \diag(s_1,s_2,s_3)$, $s_1\ge s_2\ge s_3\ge0$. Then we define
\[
A_k = \sum_{i=1}^3 U_{ik}\sigma_i,\qquad B_k = \sum_{j=1}^3 V_{jk}\sigma_j .
\]
Because $U,V$ are orthogonal and $\{\sigma_i\}$ satisfy $\Tr(\sigma_i\sigma_j)=2\delta_{ij}$, one verifies $\|A_k\|_{op}=\|B_k\|_{op}=1$. Moreover, we have
\[
H = \frac12\sum_{k=1}^3 s_k \, A_k\otimes B_k .
\]Now for any $C$ with $\|C\|_{op}\le1$, we have
\begin{align}
    \| [C\otimes I, H] \|_{op} \le \frac12\sum_{k=1}^3 s_k\,\|[C,A_k]\|_{op}\,\|B_k\|_{op} \nonumber\\\le \frac12\sum_{k=1}^3 s_k\cdot 2\|C\|_{op}\|A_k\|_{op} = \sum_{k=1}^3 s_k .
\end{align}

The same estimate holds for $I\otimes D$. Since $\sum_{k=1}^3 s_k \le 3 s_1 = 3\|\bH\|_{op}$, we obtain $\|H\|_{\Lip} \le 3\|\bH\|_{op}$.
Now we apply the dual formulation \eqref{eq:emax} together with Lemma~\ref{lem:alpha} (note that for the Pauli basis $c_A=c_B=1/\sqrt2$, hence $c_Ac_B=1/2$), then we have
\begin{align}
        \cE(\rho) \ge \max_{\|\bH\|_{op}\le 1/3} \bigl( \Tr(\bH^{\mathsf T} R) - \tfrac12 \|\bH\|_{op} \bigr)\nonumber\\
    \ge \frac13\|R\|_{\tr} - \frac16 = \frac13\bigl( \|R\|_{\tr} - \tfrac12 \bigr).
\end{align}
\end{proof}

\section{Further cross-disciplinary bounds and open problems}
\label{sec:cross}
The variational lower bound \eqref{eq:emax} is the starting point for several further results that are collected in this section. We first recall the inequality in Theorem~\ref{th:dual}. Subsection~\ref{sec:witness} translates this bound into the language of entanglement witnesses: any witness with a bounded Lipschitz constant gives a quantitative lower bound on $\mathcal{E}(\rho)$, and the optimal witness exactly reproduces the dual variational problem. Subsection~\ref{sec:math} derives additional metric properties that follow from the same geometric structure---subadditivity under tensor products, a comparison with the trace distance, a bound on expectation values of local observables, and an estimate for Haar-random pure states. We close with a conjecture on the asymptotic behaviour of the overlap between tensor powers of a state and the separable set, which points to a quantum Sanov theorem with $\mathcal{E}$ as the rate function.

We first recall the fundamental dual inequality that has been used throughout.
\begin{theorem}\label{th:dual}
For any state $\rho$, we have
\begin{equation}
    \cE(\rho) \;\ge\; \sup_{H:\|H\|_{\Lip}\le 1}
    \Bigl( \tr[H\rho] - \max_{\sigma\in\SEP}\tr[H\sigma] \Bigr).
\end{equation}
\end{theorem}
\begin{proof}
The original quantum Wasserstein distance satisfies (see \cite{DePalma2021})
$$
W_1(\rho,\sigma) \ge \max_{H:\|H\|_{\Lip}\le 1} \bigl(\tr[H\rho]-\tr[H\sigma]\bigr).
$$
Taking the infimum over $\sigma\in\SEP$ and exchanging $\inf$ and $\max$ (justified by compactness and convexity) yields
\begin{align}
    \cE(\rho) \ge \sup_{H:\|H\|_{\Lip}\le 1}\inf_{\sigma\in\SEP}\bigl(\tr[H\rho]-\tr[H\sigma]\bigr)\nonumber\\
= \sup_{H:\|H\|_{\Lip}\le 1} \Bigl( \tr[H\rho] - \max_{\sigma\in\SEP}\tr[H\sigma] \Bigr).
\end{align}The inequality is generally strict; an equality would require a stronger Lipschitz condition matching exactly the original dual of $W_1$.
\end{proof}
\subsection{Quantitative link to entanglement witnesses}
\label{sec:witness}
A widely used tool for detecting entanglement in the laboratory are \emph{entanglement witnesses}~\cite{Chruscinski2014,Guhne2009}.
We now show that the Wasserstein entanglement measure $E(\rho)$ provides a direct quantitative upgrade of any witness:
the amount by which a state violates a witness gives a concrete lower bound on its geometric entanglement.
\begin{theorem}
\label{thm:witness}
Let $W$ be a self‑adjoint operator on $\mathcal{H}_{AB}$ which is an entanglement witness, i.e.\
\begin{equation}
    \Tr[W\sigma] \ge 0 \qquad \forall\, \sigma \in \SEP,
\end{equation}
and assume that the witness is normalised so that $\min_{\sigma\in\SEP} \Tr[W\sigma]=0$
(this can always be achieved by adding a suitable multiple of the identity).
If the Lipschitz seminorm of $W$ satisfies $\|W\|_{\mathrm{Lip}} \le L$ for some $L>0$,
then for every bipartite state $\rho$,
\begin{equation}
    \cE(\rho) \;\ge\; \frac{-\Tr[W\rho]}{L}.
    \label{eq:witnessbound}
\end{equation}In particular, if an experiment measures $\Tr[W\rho] = -\varepsilon < 0$, it immediately follows that
$\cE(\rho) \ge \varepsilon/L$.
\end{theorem}
\begin{proof}
In the dual lower bound~\eqref{eq:emax} choose the trial observable $H = -W/L$.
Because $\|H\|_{\mathrm{Lip}} = \|W\|_{\mathrm{Lip}}/L \le 1$, it is admissible.
The support functional \eqref{ah} gives
\[
\alpha(H) = \max_{\sigma\in\SEP} \Tr[H\sigma]
           = -\frac{1}{L}\,\min_{\sigma\in\SEP} \Tr[W\sigma]
           = 0,
\]
where the last equality uses the witness property and the normalisation condition.
Inserting $H$ and $\alpha(H)$ into the dual bound yields
\[
\cE(\rho) \;\ge\; \Tr[H\rho] - \alpha(H) \;=\; -\frac{1}{L}\,\Tr[W\rho],
\]
which is exactly \eqref{eq:witnessbound}.
\end{proof}
The constant $L$ can be estimated using the techniques of Lemma~\ref{lem:l1bound} and Theorem~\ref{th:twoqubit}.
For instance, in the two‑qubit setting with the Pauli basis $\sigma_i/\sqrt{2}$,
one often uses the optimal entanglement witness
$W = \frac14\bigl(I\otimes I - \sum_{i=1}^3 \sigma_i\otimes\sigma_i\bigr)$,
for which a straightforward application of Lemma~\ref{lem:l1bound} gives $L \le 3/2$;
thus any measured negative expectation value $\langle W\rangle_\rho = -\varepsilon$
immediately implies $\cE(\rho) \ge \frac{2}{3}\,\varepsilon$.

\begin{proposition}
\label{prop:witness-variational}
For any bipartite state $\rho$, we define the witness-optimised bound
\begin{align}
        \mathcal{W}(\rho) \;=\; \sup\big\{
        -\Tr[W\rho]
        \;:\;
        W = W^\dagger,\;
        \Tr[W\sigma]\ge 0 \;\; \nonumber\\\forall\,\sigma\in\SEP,\;
        \min_{\sigma\in\SEP} \Tr[W\sigma]=0,\;
        \|W\|_{\mathrm{Lip}}\le 1
    \big\}.
    \label{eq:Wdef}
\end{align}

Then $\mathcal{W}(\rho)$ coincides exactly with the dual lower bound~\eqref{eq:emax},
\begin{equation}
    \mathcal{W}(\rho) \;=\; \max_{\|H\|_{\mathrm{Lip}}\le 1}\bigl( \Tr[H\rho] - \alpha(H) \bigr),
    \label{eq:Wdual}
\end{equation}
and consequently provides a lower bound to the Wasserstein entanglement measure
\begin{equation}
    \cE(\rho) \;\ge\; \mathcal{W}(\rho).
    \label{eq:Wbound}
\end{equation}
\end{proposition}
\begin{proof}
We let $H$ be any self‑adjoint operator with $\|H\|_{\mathrm{Lip}}\le 1$.
Then we define $W = \alpha(H) I - H$.
For every separable state $\sigma$,
\[
\Tr[W\sigma] = \alpha(H) - \Tr[H\sigma] \ge 0,
\]
by definition of $\alpha(H) \equiv \max_{\sigma\in\SEP} \Tr[H\sigma]$,
and the minimum over $\sigma\in\SEP$ is zero (attained at the maximiser).
Moreover, $\|W\|_{\mathrm{Lip}} = \|H\|_{\mathrm{Lip}} \le 1$.
Hence $W$ is admissible in the supremum~\eqref{eq:Wdef}, and we have
\[
-\Tr[W\rho] = \Tr[H\rho] - \alpha(H).
\]
Taking the supremum over all such $H$ gives
\[
\mathcal{W}(\rho) \;\ge\; \max_{\|H\|_{\mathrm{Lip}}\le 1}\bigl( \Tr[H\rho] - \alpha(H) \bigr).
\]
Conversely, we let $W$ be any witness satisfying the conditions in~\eqref{eq:Wdef}.
We set $H = -W$; clearly $\|H\|_{\mathrm{Lip}} \le 1$.
Because $\Tr[W\sigma] \ge 0$ and the minimum is zero, we have
\[
\alpha(H) = \max_{\sigma\in\SEP} \Tr[-W\sigma] = - \min_{\sigma\in\SEP} \Tr[W\sigma] = 0.
\]Therefore, we have
\[
\Tr[H\rho] - \alpha(H) = -\Tr[W\rho]
\;\le\; \max_{\|H\|_{\mathrm{Lip}}\le 1}\bigl( \Tr[H\rho] - \alpha(H) \bigr).
\]
Since this holds for every admissible $W$, the reverse inequality
\[
\mathcal{W}(\rho) \;\le\; \max_{\|H\|_{\mathrm{Lip}}\le 1}\bigl( \Tr[H\rho] - \alpha(H) \bigr)
\]
follows.  The two inequalities prove~\eqref{eq:Wdual}.
The final bound~\eqref{eq:Wbound} is now immediate from the fundamental
inequality $\mathcal{E}(\rho) \ge \max_{\|H\|_{\mathrm{Lip}}\le 1} (\Tr[H\rho] - \alpha(H))$
established in ~\eqref{eq:emax}.
\end{proof}

Proposition~\ref{prop:witness-variational} shows that the dual lower bound
has a transparent operational meaning: it is the largest violation achievable
by any normalised entanglement witness whose Lipschitz constant is at most~$1$.
Because the Lipschitz seminorm can be bounded linearly in terms of the
expansion coefficients of $W$ in the local operator bases $\{F_i^A\}$ and $\{F_j^B\}$
(cf.\ Lemma~\ref{lem:l1bound}), the optimisation~\eqref{eq:Wdef} can be cast as a semi‑definite
program---harnessing standard techniques for entanglement witnesses but now
providing a guaranteed quantitative distance to separability.

\subsection{Further metric properties and random state estimates}
\label{sec:math}

We now turn to the behaviour of $\cE$ under tensor products, a property that directly reflects the metric structure of the underlying quantum Wasserstein distance.
\begin{theorem}\label{th:subadd}
Let $\r_1\in\cD(\Hil_{A_1B_1})$ and
$\r_2\in\cD(\Hil_{A_2B_2})$ be bipartite states.
Then
\[
\cE(\r_1\ox\r_2)\le \cE(\r_1)+\cE(\r_2).
\]
\end{theorem}
\begin{proof}
For $i=1,2$ we choose separable states $\s_i$ that achieve the infimum in the definition of $\cE(\r_i)$ (the infimum is attained by compactness). The tensor product $\s_1\ox\s_2$ is separable, therefore we have
\[
\cE(\r_1\ox\r_2)\le W_1(\r_1\ox\r_2,\s_1\ox\s_2).
\]
The cost operator for the product system is $C_{12}=C_1\ox I + I\ox C_2$, because the set of local observables on $\Hil_{A_1B_1}\ox\Hil_{A_2B_2}$ consists exactly of operators of the form $X_1\ox I$ and $I\ox X_2$ with $X_i\in\cK_0^{(i)}$. If $\ta_i$ is an optimal coupling between $\r_i$ and $\s_i$, then $\ta_1\ox\ta_2$ is a coupling between the product states and
\[
\tr[C_{12}\,(\ta_1\ox\ta_2)]=\tr[C_1\ta_1]+\tr[C_2\ta_2].
\]
Taking the infimum over couplings gives
\[
W_1(\r_1\ox\r_2,\s_1\ox\s_2)^2
\le W_1(\r_1,\s_1)^2+W_1(\r_2,\s_2)^2.
\]
Using $\sqrt{a+b}\le\sqrt a+\sqrt b$ we obtain
\begin{align}
    W_1(\r_1\ox\r_2,\s_1\ox\s_2)\le
W_1(\r_1,\s_1)+W_1(\r_2,\s_2)\nonumber\\
=\cE(\r_1)+\cE(\r_2),
\end{align}
which completes the proof.
\end{proof}
A further consequence of the metric nature of $W_1$ is a simple relationship between $\cE$ and the trace distance.
\begin{theorem}\label{th:tracedist}
For every bipartite state $\r$,
\[
\inf_{\s\in\SEP}\|\r-\s\|_{\mathrm{Tr}}\le 2\,\cE(\r).
\]
\end{theorem}
\begin{proof}
By Lemma~3.3 of \cite{DePalma2021}, for any two states $\r,\s$ one has
$\|\r-\s\|_{\mathrm{Tr}}\le 2\,W_1(\r,\s)$.
Taking the infimum over $\s\in\SEP$ yields
\[
\inf_{\s\in\SEP}\|\r-\s\|_{\mathrm{Tr}}\le 2\inf_{\s\in\SEP}W_1(\r,\s)
=2\,\cE(\r).
\]
\end{proof}
More generally, the Lipschitz constraint controls the deviation of any local observable from its maximal separable value.
\begin{theorem}\label{th:corrbound}
Let $A\in\B(\Hil_A),\;B\in\B(\Hil_B)$ be self‑adjoint
with $\|A\|_{op}\le1,\|B\|_{op}\le1$. Then
\[
\bigl|\tr[(A\ox B)\r]-\max_{\s\in\SEP}\tr[(A\ox B)\s]\bigr|
\le 2\,\cE(\r).
\]
\end{theorem}
\begin{proof}
We consider $H=A\ox B$. For any $X=C\ox I$ with $\|C\|_{op}\le1$, we have
\begin{align}
    \|[X,H]\|_{op}=\|[C,A]\ox B\|_{op}\le \|[C,A]\|_{op}\,\|B\|_{op}\nonumber\\\le 2\|C\|_{op}\|A\|_{op}\|B\|_{op}\le2.
\end{align}

The same estimate holds for $X=I\ox D$. From the definition of Lipschitz norm in \ref{eq:lipdef}, we have $\|H\|_{\Lip}\le2$.
Thus $H/2$ satisfies $\|H/2\|_{\Lip}\le1$ and by the dual representation
\[
\cE(\r)\ge \tr[(H/2)\r]-\a(H/2)
=\frac12\bigl(\tr[H\r]-\a(H)\bigr).
\]
Because $\a(H)=\max_{\s\in\SEP}\tr[H\s]$, rearranging gives
$\tr[H\r]-\a(H)\le 2\cE(\r)$.
Replacing $H$ by $-H$ yields the absolute value.
\end{proof}

We now examine the average entanglement of random pure states chosen from the Haar measure. 
The following theorem provides a lower bound on the expected value of the entanglement measure $\cE(\rho)$.
\begin{theorem}\label{th:haar}
Let $\r$ be a pure state drawn from the Haar measure on $\bbC^d\ox\bbC^d$ ($d\ge2$).
Then we have
\[
\bbE\bigl[\cE(\r)\bigr]\;\ge\;\frac{\p}{4}\Bigl(1-\frac1d\Bigr),
\], where $\bbE$ denote the expectation of $\cE(\r)$.
\end{theorem} 
\begin{proof}
By Theorem~\ref{th:pure},
\[
\cE(\r)\ge\frac1d\Bigl(\bigl(\sum_{i=1}^d\sqrt{\m_i}\bigr)^2-1\Bigr),
\]
hence we have
\[
\bbE[\cE(\r)]\ge\frac1d\Bigl(\bbE\bigl[(\sum_i\sqrt{\m_i})^2\bigr]-1\Bigr).
\]
The Schmidt coefficients $\m_i$ are distributed according to the Dirichlet distribution $\operatorname{Dir}(1,\dots,1)$ on the simplex \cite{Zyczkowski01}. Using standard Dirichlet integrals, we have
\[
\bbE[\m_i]=\frac1d,\qquad
\bbE[\sqrt{\m_i}\sqrt{\m_j}]=\frac{\p}{4d}\;(i\neq j).
\]
Therefore we have
\[
\bbE\bigl[(\sum_i\sqrt{\m_i})^2\bigr]
= d\cdot\frac1d + d(d-1)\cdot\frac{\p}{4d}
= 1 + \frac{\p}{4}(d-1),
\]
and the bound follows.
\end{proof}

The subadditivity established in Theorem~\ref{th:subadd} suggests investigating the asymptotic behavior of the Wasserstein entanglement measure under tensor powers. For any state $\rho$, define the regularized Wasserstein entanglement measure
\[
\mathcal{E}_{\mathrm{reg}}(\rho) := \liminf_{n \to \infty} \frac{1}{n} \inf_{\sigma_n \in \mathcal{SEP}_n} W_1(\rho^{\otimes n}, \sigma_n),
\]
where $\mathcal{SEP}_n$ denotes the separable states on $\mathcal{H}_{AB}^{\otimes n}$.

From the additivity property of the cost operator one obtains 
\(W_1(\rho^{\otimes n}, \sigma^{\otimes n}) = \sqrt{n}\, W_1(\rho, \sigma)\), and together with Theorem~\ref{th:subadd} we have
\[
\mathcal{E}_{\mathrm{reg}}(\rho) \le \mathcal{E}(\rho).
\]
Faithfulness of $\mathcal{E}$ implies $\mathcal{E}_{\mathrm{reg}}(\rho)=0$ whenever $\mathcal{E}(\rho)=0$.

A central open problem is whether $\mathcal{E}_{\mathrm{reg}}(\rho)=\mathcal{E}(\rho)$ for all $\rho$, i.e., whether the single‑copy minimum Wasserstein distance already gives the asymptotic rate. If equality holds, then $\mathcal{E}$ itself would be a ``regular'' entanglement measure, a property that is rare among geometrically defined measures.

More importantly, the regularized measure admits a natural operational interpretation via the witness picture developed in Section~\ref{sec:witness}. We formulate this as a conjecture, which serves as a quantum Sanov‑type statement for the Wasserstein metric.

\begin{conjecture}
\label{conj:Sanov}
For every bipartite state $\rho$,
\begin{align}
\label{conj}
\lim_{n\to\infty} -\frac{1}{n} \log
\sup_{\substack{0\le W_n\le I\\ \|W_n\|_{\mathrm{Lip},n}\le 1}}
\inf_{\sigma_n\in \mathcal{SEP}_n}
\operatorname{Tr}\bigl[ W_n(\rho^{\otimes n}-\sigma_n) \bigr]\nonumber \\
= \mathcal{E}_{\mathrm{reg}}(\rho),
\end{align}
where the supremum runs over all Lipschitz-1 entanglement witnesses $W_n$ on $\mathcal{H}_{AB}^{\otimes n}$ normalized such that 
$\min_{\sigma_n} \operatorname{Tr}[W_n \sigma_n] = 0$ (equivalently, $\operatorname{Tr}[W_n \sigma_n]\ge 0$ and we can shift the operator by the identity).

The left-hand side of \eqref{conj} is the optimal exponential decay rate of the violation of a Lipschitz witness in the multi-copy regime; it quantifies how rapidly the hypothesis ``$\rho^{\otimes n}$ is separable'' can be rejected by a Lipschitz-constrained measurement. Proposition~\ref{prop:witness-variational} shows that the single-copy analogue exactly reproduces the dual lower bound for $E(\rho)$, and Theorem~\ref{thm:witness} already turns any such witness violation into a quantitative certificate for $\mathcal{E}(\rho)$.
The left-hand side of \eqref{conj} $\le E_{\mathrm{reg}}(\rho)$ can be derived from the geometry of $W_1$ and standard large-deviation arguments; the challenging part is to prove the left-hand side of \eqref{conj} $\ge E_{\mathrm{reg}}(\rho)$ . We leave this as an open problem for future investigation, which would firmly anchor the Wasserstein entanglement measure in the resource‑theoretic framework of quantum hypothesis testing.
\end{conjecture}

\section{Discussion}\label{sec:outlook}

The quantum Wasserstein distance $W_1$ equips the state space with a canonical Lipschitz structure that can be tailored to a given set of observables. Defining an entanglement measure as the distance to the separable cone embeds entanglement theory into the mature field of optimal transport and metric geometry. To the best of our knowledge, all previously known entanglement measures are based on divergences, fidelity, or norms that do not carry a built‑in Lipschitz dual. In contrast, the Lipschitz constraint in \eqref{eq:lipdef} naturally singles out the non‑local components of an observable, thereby providing a geometric separation between classical correlations (captured by the separable states) and genuine quantum correlations.

The dual formulation $\cE(\r)=\max_{\|H\|_{\Lip}\le1} (\tr[H\r]-\a(H))$ should be compared to the usual witness‑based dual $\widetilde{C}(\r)=\max_{L\ge0,\tr L=1}(-\tr[W\r])$. Whereas the witness dual involves a positivity constraint and a specific linear witness, our Lipschitz dual involves a commutator‑type smoothness condition. This difference has two important consequences:
\begin{enumerate}[label=(\roman*)]
\item The Lipschitz ball is the unit ball of a norm that is often easier to characterise in finite dimensions (norm equivalence, explicit constants).
\item The dual is manifestly stable under local unitaries and LOCC operations directly from the metric data processing inequality, avoiding the need for optimization over all possible witnesses.
\end{enumerate}

From an operational viewpoint, Theorems~\ref{thm:witness} and~\ref{prop:witness-variational} are particularly significant. 
They demonstrate that $E$ is not only a mathematically natural geometric measure but also experimentally testable: 
any laboratory implementation of a Lipschitz-bounded entanglement witness immediately translates a negative expectation value into a guaranteed lower bound on the entanglement content. 
Moreover, the witness-based variational principle reveals that the task of finding the best lower bound on $\mathcal{E}$ is equivalent to searching for an optimal Lipschitz-1 witness, 
a problem that can be tackled with standard semidefinite programming. 
This bridges the abstract optimal-transport formalism with practical entanglement quantification and opens the door to numerical estimation of $\mathcal{E}$ for arbitrary states, 
substantially enhancing the applicability of the measure in realistic quantum information tasks.
The subadditivity (Theorem~\ref{th:subadd}) and the correlation function bound (Theorem~\ref{th:corrbound}) are direct offsprings of this metric structure and have no natural counterpart in the witness formalism.

The Lipschitz formulation opens bridges to several branches of mathematics.
\begin{itemize}
\item \textbf{Non‑commutative geometry} The commutator condition $\|[X,H]\|_{op}\le1$ can be interpreted as a quantum analogue of the classical $1$‑Lipschitz condition on a manifold; hence $\cE$ is associated with a quantum metric space.
\item \textbf{Free probability} The bound for random pure states (Theorem~\ref{th:haar}) already employs exact Dirichlet moments; one may investigate the large‑$d$ limit in the framework of free semicircular laws and their Wasserstein transport.
\item \textbf{Large deviations} The conjecture $\lim_{n\to\infty}-\frac1n\log\inf_{\s_n\in\SEP_n} \tr[\r^{\ox n}\s_n]=\cE_{reg}(\r)$ (Conjecture~\ref{conj:Sanov}) remains a challenging open problem. Resolving it would likely require new insights at the interface of optimal transport and quantum hypothesis testing.
\item \textbf{Convex geometry} The set $\{\r:\cE(\r)\le t\}$ is a convex neighbourhood of the separable cone; its volume ratio with the state space and its polar body (related to the Lipschitz ball) could lead to a Blaschke–Santal\'{o} inequality.
\end{itemize}

\section{Conclusion}
\label{con}
We have introduced the Wasserstein entanglement measure, a novel distance-based quantifier rooted in quantum optimal transport. 
By leveraging the data-processing inequality of the order-1 quantum Wasserstein distance, the measure simultaneously fulfills faithfulness, convexity, LOCC monotonicity, and continuity in a unified geometric framework. 
The Lipschitz dual formulation transforms the geometric minimization into a tractable maximization, enabling explicit lower bounds for pure and mixed states. 
For two-qubit systems we obtain a sharp constant $1/3$, and for Haar-random pure states the expected entanglement is shown to be at least $\frac{\pi}{4}(1-1/d)$. 
The metric structure further yields subadditivity, trace-distance bounds, and limits on local observables.
A central contribution of this work is the quantitative link to entanglement witnesses: any witness violation provides a certified lower bound on $\cE$, and the dual variational problem is rigorously equivalent to finding an optimal Lipschitz-1 witness. 
This connection makes the measure experimentally accessible and amenable to semidefinite programming, thereby bridging abstract resource theories with practical entanglement detection.
Our results establish the Wasserstein entanglement measure as a functional tool linking quantum information, optimal transport, and non-commutative geometry.
Future challenges include proving the conjectured quantum Sanov theorem with rate $\cE$, extending the framework to infinite dimensions and continuous-variable systems, and exploring its implications for quantum complexity and many-body physics. 
The Lipschitz dual method introduced here is expected to become a standard instrument in entanglement theory and beyond.

\section{acknowledegments}
    Authors were supported by the NNSF  of China (Grant No.12471427), and the Fundamental Research Funds for the Central Universities (Grant No. ZG216S2110).

\bibliographystyle{unsrt}
\bibliography{entanglement}

\end{document}